\documentclass[12pt]{article}  
\usepackage{amsmath}
\usepackage{amsthm}
\usepackage{amsfonts}
\usepackage{amssymb}
\usepackage{url}
\usepackage{graphicx}
\usepackage{wrapfig}
\usepackage{lscape}
\usepackage{rotating}
\usepackage{epstopdf}
\usepackage{setspace} 
\usepackage[font=footnotesize]{caption}
\usepackage{scicite}

\newenvironment{sciabstract}{%
\begin{quote} \bf}
{\end{quote}}

\newcounter{lastnote}

\newtheorem{theorem}{Theorem}
\newtheorem{lemma}[theorem]{Lemma}

\title{How structurally stable are \\global socioeconomic systems?}
\author{Serguei Saavedra$^{1}$\footnote{To whom correspondence should be addressed. E-mail: serguei.saavedra@ebd.csic.es} \footnote{These authors contributed equally to this work} , Rudolf P. Rohr$^{1,2}$\footnotemark[\value{footnote}] , \\Luis J. Gilarranz$^1$, and Jordi Bascompte$^1$ 
\\        
\\$^1$Integrative Ecology Group \\ Estaci\'on Biol\'ogica de Do\~nana (EBD-CSIC) \\ Calle Am\'erico Vespucio s/n \\ E-41092 Sevilla, Spain\\
\\$^2$Unit of Ecology and Evolution \\ Department of Biology, University of Fribourg \\ Chemin du Mus\'{e}e 10 \\ CH-1700 Fribourg, Switzerland}

\date{}

\begin{document}
\maketitle

\vspace{0.4 in}
 

\newpage


\begin{sciabstract}
The stability analysis of socioeconomic systems has been centered on answering whether small perturbations when a system is in a given quantitative state will push the system permanently to a different quantitative state. However, typically the quantitative state of socioeconomic systems is subject to constant change. Therefore, a key stability question that has been under-investigated is how strong the conditions of a system itself can change before the system moves to a qualitatively different behavior, i.e., how structurally stable the systems is. Here, we introduce a framework to investigate the structural stability of socioeconomic systems formed by the network of interactions among agents competing for resources. We measure the structural stability of the system as the range of conditions in the distribution and availability of resources compatible with the qualitative behavior in which all the constituent agents can be self-sustained across time. To illustrate our framework, we study an empirical representation of the global socioeconomic system formed by countries sharing and competing for multinational companies used as proxy for resources. We demonstrate that the structural stability of the system is inversely associated with the level of competition and the level of heterogeneity in the distribution of resources. Importantly, we show that the qualitative behavior of the observed global socioeconomic system is highly sensitive to changes in the distribution of resources. We believe this work provides a methodological basis to develop sustainable strategies for socioeconomic systems subject to constantly changing conditions.
\end{sciabstract}

\clearpage

\baselineskip=8.5mm

\section{Introduction}

The stability of socioeconomic systems is repeatedly challenged as a consequence of the rapidly varying environmental, socioeconomic, and technological conditions \cite{Scheffer,Saavedra2,May09}. Financial crisis, national bailouts, and job losses are just a few examples of instability in these systems \cite{Scheffer,May09}. The stability analysis of socioeconomic systems has been centered on understanding whether small perturbations when a system is in a given quantitative state will push the system permanently to a different quantitative state \cite{May09,Schweitzer,Haldane,Battiston,Barzel}. This analysis is known as dynamical stability \cite{Strogatz2001}. Importantly, dynamical stability has increased our understanding on the susceptibility of socioeconomic systems to propagate specific perturbations \cite{May09,Schweitzer,Haldane,Battiston,Barzel}. However, as the quantitative state of socioeconomic systems is coevolving with the rapidly changing distribution and availability of resources, economists are not only interested in a particular steady state, but also in whether there is a family of quantitative states that can guarantee the sustainability of these systems \cite{Keynes,Tinberg,Sargent,Lucas,Arrow}. This indicates that a yet prevailing question about socioeconomic systems is how much variation can a system stand without being pushed out of a qualitative stable behavior \cite{Thom,Vandermeer,Saavedra2}.

To address the above question, we apply the concept of structural stability to socioeconomic systems. We adopt a modified definition of structural stability \cite{Thom,Alberch85,Rohr2}, in which a system is more structurally stable if it has a larger range of conditions compatible with a given qualitative stable state. Here, we explore the structural stability of a general resource-competition system by considering a qualitative behavior in which all its constituent agents have a positive and stable steady state. We choose a positive stable steady state as a potential indicator of an agent that can be self-sustained across time without the need of external inputs. Therefore, the question is: how big is the parameter space in the system compatible with this positive stable steady state? The larger the range of parameter space compatible with a positive stable steady state of all agents, the larger the structural stability of the system will be. 

To illustrate our framework, we study an empirical representation of the global socioeconomic system formed by the network of interactions among countries (agents) competing for multinational companies (as proxy for resources such as investment, technological innovations, and employment). We investigate the range of conditions leading to the preferred qualitative behavior and the mechanisms modulating that range.
 
\section{Materials and Methods}

\subsection{Competition Network}

Our global socioeconomic system is represented by the network of interactions among countries competing for resources. Following economic theory \cite{Keynes,Tinberg,Sargent,Lucas,Arrow}, we focus on three main resources for economic growth: private investment, technological innovations, and employment. We use the 50-richest multinational companies in the world as proxy for these resources. We acknowledge that there can be other representations of these resources that might be important or useful. The list of these companies is taken from the 2013 Fortune Global 500 list. The total revenue of these companies is about $30\%$ of the world's gross domestic product (GDP). We consider that a country utilizes a resource (multinational company) only when the company has employees in that country. Note that we do not have quantitative data on the number of employees. This information is collected from each official company's website in 2013. We focus on 150 countries with at least one million inhabitants. This dataset is provided in the Data Supplement. 

The competition dynamics of socioeconomic systems have been studied using either static equilibrium models\cite{Arrow,Sargent} or exponential growth models\cite{Lucas,Solow,Romer} with no explicit interactions among agents. This has precluded the analysis of socioeconomic systems as potential systems with nonlinear dynamics emerging from collective phenomena and regulated by the network of interactions among their individual agents \cite{Strogatz2001,Anderson,Costanza}. To incorporate these interactions, we propose to model the socioeconomic system as an inter-agent resource-competition network. To define our competition network, first we generate a resource-agent system composed of $N$ agents (countries) and $R$ resources (companies). This system is represented as a bipartite network made of two set of nodes, the agents and their resources. A binary link is drawn between an agent $i$ and a resource $k$ if the agent uses the given resource (See Fig. 1a for a graphical representation). Second, we transform the previously generated resource-agent system into an inter-agent resource-competition network. This competition network is characterized by a symmetric matrix $\boldsymbol{\beta}$ of size $N \times N$, called the competition matrix. The elements of the competition matrix $\beta_{ij}$ are a function of the number of shared resources between agents (See Fig. 1b for a graphical representation). 

\subsection{Dynamics of the competition network}

Formally, we describe the dynamics of our inter-agent resource-competition network by a general Lotka-Volterra model given by the following set of ordinary differential equations \cite{MacArthur,Case2000}.

\begin{equation} \label{equ:ode}
	\frac{dN_{i}}{dt} = \frac{r_{i}}{K_{i}} N_{i}(K_{i} - \sum_{j} \beta_{ij} N_{j}),
\end{equation}
where $N_i \geq 0$ denotes the abundance of the agent $i$ (e.g., the wealth of a country), $r_i >0$ is the growth rate of the agent $i$, and $K_i > 0$ is the carrying capacity of agent $i$. The elements $\beta_{ij}$ correspond to the per capita effect of agent $j$ on the abundance of agent $i$. These elements are given by the values extracted from the competition matrix. By convention and without loss of generality, we set the intra-agent resource-competition to one ($\beta_{ii} = 1$). The off-diagonal elements are set to $\beta_{ij} = \mu \cdot c_{ij}$ ($i \neq j$), where $c_{ij}$ is the number of shared resources between agents $i$ and $j$, and $\mu$ is the general level of global competition in the system ($\mu \geq 0$). This model description emulates current economic thinking on the existence of limited resources and nonlinear dynamics of socioeconomic systems \cite{Anderson,Costanza}. 

In the simple scenario where agents do not compete among them, i.e., when the inter-agent competition is set to zero ($\beta_{ij}=0$ for $i \neq j$), the carrying capacity alone dictates the steady state of the system $N^*_i = K_i$. Moreover, under the condition that $K_i > 0$, it can be mathematically proven that this steady state is globally stable, and that the growth rate of agents only modulates the velocity at which each agent reaches its own carrying capacity. This means that the qualitative behavior in which all agents have a positive and constant abundance ($N_i^*>0$)---what we refer to as the positive stable steady state---can only be possible if the carrying capacity of all agents is also positive ($K_i>0$). See Appendix A for mathematical details.

In the more complex scenario where agents do compete among them for resources, the steady state of the system is function of both the carrying capacity and the competition matrix. It can be mathematically proven that if all eigenvalues of the competition matrix  $\boldsymbol{\beta}$ are positive (they are real because this matrix is symmetric) and if there exists a positive steady state for all agents ($N_i^* > 0$), then this positive steady state is a global attractor in the strictly positive quadrant of the state space \cite{Goh}. Moreover, it can also be mathematically proven that for any vector of carrying capacity $K_i >0$ (keeping the positive eigenvalue condition on the competition matrix), the dynamical system will converge to a unique equilibrium point $N_i^* \geq 0$, where the state of either all or only a few of the agents is positive. See Appendix A for mathematical details. 

The condition of global stability (i.e., eigenvalues of the competition matrix $\boldsymbol\beta$ are all positive) only holds when $\mu$ is below a critical value $\hat{\mu}$ at which one eigenvalue of the competition matrix is equal to zero (see Appendix A for further details). A limitation of the level of global competition $\mu$ is that it has the same units as the competition elements $\beta_{ij}$, and it is not possible to compare this level across different competition matrices. To address this problem, we recast this level by a unit-free indicator of the level of global competition ($\rho$). It is defined as $\rho = \frac{\lambda_1-1}{N-1}$, where $N$ is the number of agents, and $\lambda_1$ is the dominant eigenvalue of the competition matrix $\boldsymbol\beta$.

To find a positive and globally stable steady state of our system, we have to solve the following linear equation $\boldsymbol K = \boldsymbol\beta \cdot \boldsymbol{N^*}$ under the constraint that $N_i^* >0$. Importantly, not all vectors $\boldsymbol K$ lead to a positive steady state. However, if we set the vector $\boldsymbol K^*$ equal to the leading eigenvector of the competition matrix $\boldsymbol\beta$---what we call the structural vector of carrying capacity---we obtain a non-trivial solution. Indeed, following the Perron-Frobenius theorem, the corresponding equilibrium point of the structural vector is non-trivial and given by $N_i^* = \frac{1}{\lambda_{1}} K_i^* > 0$, where $\lambda_1$ is the leading eigenvalue of $\boldsymbol\beta$. 

\subsection{Structural stability of the competition network}

We study the structural stability of our global socioeconomic system by measuring how much variation the resource-competition system can stand without being pushed out of the positive stable steady state. We explore the range in the parameter space of carrying capacities that leads the system to the global stable equilibrium point of Equation \ref{equ:ode} in which all agents have a positive steady state ($N^*_i>0$). To quantify this rage, we measure how big the deviations are from the structural vector compatible with a positive stable steady state of all agents. These deviations are quantified by $ \eta = \frac{1 - \cos^2(\theta)}{\cos^2(\theta)}$, where $\theta$ is the angle between the structural vector $\boldsymbol{K}^*$ and any other parameterization---vector $\boldsymbol{K}$---that can be used as proxy for different conditions in the system, such as different availability of resources.

Indeed, the range of conditions compatible with our definition of positive stable steady state is centered on the structural vector $\boldsymbol{K^*}$. This is shown by the following derivation. To find a non-trivial equilibrium point $N_i^* >0$, we can link the deviation $\eta$ with the indicator of global competition $\rho$ by satisfying the inequality $\eta < \frac{1-\rho}{(N-1)\rho + 1}$ \cite{Bastolla1}. From this inequality, we can see that the lower the level of global competition $\rho$, the lower the collinearity between the structural vector and any other vector and, in turn, the wider the conditions for having the solution $N_i^*>0$. This provides a good indication that the structural vector is the symmetry axis of the hypervolume of the range where the stable solution $N_i^*>0$ is positive.

\section{Results}

\subsection{Validation of model parameterization}

To validate our model parameterization, we investigate whether the positive and globally stable steady state $N^*_i>0$ (given by the structural vector of carrying capacities) is aligned with key macroeconomic indicators of our global socioeconomic system. Recall that the steady state defined by the structural vector is computed as $N_i^* = \frac{1}{\lambda_{1}} K_i^* > 0$, where $\lambda_1$ is the leading eigenvalue of $\boldsymbol\beta$. Interestingly, we find a strong and positive Spearman rank correlation ($r=0.88$, $p<0.001$) between the equilibrium point and countries' GDP (Fig. 2a). The same positive correlation is observed between the number of resources and the GDP of a country, suggesting that wealth is strongly associated with the distribution of resources in our system.

We further test the alignment between the observed resource-competition network and model parameterization by generating new equilibrium points calculated using the structural vector of alternative competition networks. These alternative networks are extracted from randomly generated resource-agent systems (see Appendix B). If these alternative resource-agent systems preserve, in expectation, the observed distribution of resources per agent, the positive correlation between GDP and new equilibrium points is also preserved. In contrast, if the alternative resource-agent systems do not preserve the observed distribution of resources, there are negligible correlations between GDP and the new equilibrium points (for an example see Fig. 2b). These results reveal that both our competition network and parameterization of carrying capacities are indeed capturing important characteristics of the distribution and availability of resources, respectively.

\subsection{Structural stability}

To study whether inter-agent competition increases or decreases the structural stability of the system, we study the effect of the global competition on the range of parameter space of carrying capacities leading to the positive stable steady state of all countries. We quantify this effect by the extent to which the deviations from the structural vector (given by the observed competition network) affect the fraction of countries that remain in a positive stable steady state ($N^*_i>0$), and whether these deviations are modulated by the level of global competition. The larger the range of parameter space compatible with a positive stable steady state of all countries, the larger the structural stability of the system will be.

We generate the deviations (range of parameters) by introducing random proportional perturbations to the structural vector $\boldsymbol{K^*}$, and quantify the deviation between the structural and the perturbed vectors of carrying capacity using the previously defined measure of deviation $\eta$. The proportional perturbations are generated by multiplying each element of the carrying capacity vector by a random number sampled from a lognormal distribution with mean zero and variance sampled uniformly within the range $[0,\dots,0.9]$. To find the corresponding fraction of countries that remain in a positive stable steady state, we simulate our dynamical model using the perturbed vectors as initial parameters $\boldsymbol K$. Simulations to find the equilibrium points are performed by integrating the system of ordinary differential equations using the Runge-Kutta method of Matlab routline ode45.

Figure 3 shows that when the deviation $\eta$ from the structural vector is small (negative on a log scale), all countries remain in a positive stable steady state (yellow/light region). However, the larger the deviation, the lower the fraction of countries that remain in this steady state. This confirms numerically that the structural vector is the center of the range of parameter space compatible with the positive stable steady state of all countries. Importantly, Figure 3 also reveals that the closer the system is to the boundary of maximum global competition ($\hat{\rho}$), the narrower the parameter space leading to a positive stable steady state of all countries, and in turn the lower the structural stability of the system. This reveals that the structural stability of the system decreases as the level of global competition among countries increases. 

Because the level of global competition ($\rho$) is a function of the resources shared among countries, it is important to know whether a redistribution of resources may increase or decrease the level of global competition and, in turn, affect the structural stability of the system. To capture these effects, we quantify the level of global competition ($\rho$) in alternative inter-agent resource-competition networks (extracted from randomly generated resource-agent systems, see Appendix B for further details) relative to the level of global competition computed from the observed inter-agent competition network ($\rho^*$). This means that an alternative competition network increases the level of competition when $\rho / \rho^*>1$, and vice versa when $\rho / \rho^*<1$.

In the case when alternative competition networks preserve, in expectation, the observed distribution of resources per countries, we find that the level of global competition increases relative to the observed network (see black symbols in Figure 4). These findings support standard macroeconomic theory \cite{Lucas,Arrow,Tinberg} that suggests that the observed characteristics of socioeconomic systems should be optimizing the present economic constraints. However, in the case when the distribution of resources per countries is not preserved, we find that the lower the heterogeneity among countries (measured by the standard deviation of resources per countries), the lower the level of competition $\rho / \rho^*<1$ and, in turn, the higher the structural stability of the system (see Fig. 4). These results reveal that the inter-agent resource-competition network is a significant factor modulating the range of conditions compatible with the positive stable steady state of all countries in the system. Moreover, our findings reveal that the structural stability of the system is inversely associated with the level of competition for resources and the heterogeneity in the distribution of resources.

\subsection{Risk assessment}

To provide further insights into the factors shaping the structural stability of the observed global socioeconomic system, we explore the risk associated with individual countries under rapid changes in the distribution and availability of resources. Following economic theory \cite{Lucas,Arrow,Tinberg}, we refer to rapid changes as the perturbations that can occur faster than the adaptation of the system to the new socioeconomic conditions. Specifically, we use a Monte Carlo approach to quantify the probability that a country remains in a positive stable steady state ($N^*_i>0$) when the system is subject to different types of perturbations. Specifically, perturbations are generated by random deviations from the structural vector of carrying capacities, different levels of global competition, and changes in the inter-agent resource-competition network. 

To explore the risk associated with rapid changes in the availability of resources, as before we introduce proportional random perturbations to the structural vector of carrying capacities, simulate the dynamical model on the observed competition network using the perturbed vectors as initial parameters $\boldsymbol{K}$, and investigate the fraction of times a country remains under a positive stable steady state as function of their number of resources. Interestingly, Figure 5a shows that the probability of remaining in a positive stable steady state is almost the same for all countries regardless of their number of resources. However, this probability decreases as the level of global competition in the system increases (see Fig. 5a), echoing our previous results at the network level.

Additionally, we explore the risk associated with rapid changes in the distribution of resources by randomly changing the inter-agent resource-competition network via the resource-agent system (see Appendix B). These changes are investigated both alone and in combination with changes in the availability of resources (i.e., perturbations to the structural vector). In general, we find that the lower the number of initial resources a country has, the lower its probability of remaining in a positive stable steady state (Figs. 5b-c). Overall, there seems to be a saturation point in the number of initial resources after which countries cannot increase any more their chances of remaining under a positive stable steady state. Importantly, these findings reveal that the qualitative behavior of the system is highly sensitive to rapid changes in the distribution of resources.

\section{Discussion}

In this paper, we have used a parsimonious model and network representation of a resource-competition system to investigate the structural stability of global socioeconomic systems.  However, the striking similarities found between model-generated and empirical characteristics suggest that this could be a promising starting point to answer how structurally stable global socioeconomic systems are. We have used the notion of structural stability to study the range of conditions compatible with the stability of a qualitative behavior in which all the constituent agents can be self-sustained across time. Because of the lack of detailed information about the empirical parameter values in the model, our results do not reveal the actual range of conditions tolerated by the observed global socioeconomic system. Yet, our results show that independently of parameter values, the higher the level of competition or the higher the inequality of resources among countries, the lower the structural stability of the system. Importantly, our findings suggest that multinational companies can be used as proxy for resources, and the sustainable behavior of global socioeconomic systems can be highly sensitive to changes in country-company interactions.

We believe our framework provides a new direction to increase our understanding on the capacity of a socioeconomic system to change and adapt. For instance, while the human population might be exponentially growing, we live constrained to a finite number of resources \cite{Costanza}. At present, we might be able to see an equally growing economic development simply because we have not reached our total carrying capacity, i.e., new resources are continuously being explored and exploited. If agents increase their carrying capacities by number or magnitude, they may also increase their total abundance or wealth. However, the positive stable steady state of all agents will depend on whether the new conditions in the system will be aligned or close enough to the corresponding structural vector of carrying capacities. The new challenges will be on how to deal with a limited number of resources under the constraints imposed by the structural vector and how to provide a desirable distribution of wealth among agents.

Our framework can also be applied to other domains such as biological systems. Indeed, ecological systems are constantly changing in response to both their internal and external pressures. For instance, the concept of structural stability has been applied to mutualistic systems to investigate whether there are some network characteristics that can increase the likelihood of species coexistence \cite{Rohr2}. The resource-competition system used in this work has been intensively used in ecology to describe the competition for resources among species \cite{MacArthur}. This suggests that our findings can also shed new light into the factors shaping the competition among predators that forage on a common set of prey, or the competition among plants for minerals, water, and sunlight. 

\section*{Appendices}
{\bf Appendix A. Mathematical derivations of the dynamical competition model.} In this appendix, we give analytical results for the dynamical system described by the set of ordinary differential equations (\ref{equ:ode}). Specifically, we study the existence of steady states, their feasibility (i.e., all agents having a strictly positive state), and their global stability. First, we prove that if the initial conditions of the dynamical system are in the positive quadrant ($\mathbb{R}^n_{\geq 0}$), then their trajectories also remain in the positive quadrant. This implies that we have to focus on the existence and stability of steady states in the positive quadrant only.

\begin{lemma}
	Consider a dynamical system given by the set of ordinary differential equations (\ref{equ:ode}) with initial conditions in the positive quadrant ($\mathbb{R}^n_{\geq 0}$), i.e., $N_i(t=0) \geq 0$. Then the trajectory of the system remains in the positive quadrant, i.e., $N_i(t) \geq 0$ for all time $t \geq 0$.
\end{lemma}

\begin{proof}
	Consider that there exists an agent $k$ and a time $T_1$ such that $N_k(t=T_1) <0$. Then as the trajectories of our dynamical system (\ref{equ:ode}) are continuous, there exists $T_0 < T_1$ such that $N_k(t=T_0) = 0$.  This implies that at the time $T_0$ the derivative of $N_k$ vanishes, i.e., $\frac{dN_k}{dt} |_{t=T_0} = 0$. Moreover, this equality is independent on the values of $N_i$ for all $i \neq k$. Therefore, we have that $N_k(t \geq T_0) = 0$, and in particular that $N_k(t=T_1) = 0$. This contradiction proves the lemma.
\end{proof}

Recall that a steady state $\boldsymbol{N^*}$ is called positive if $N_i^*>0$ for all agents $i$. Any positive steady state is by definition the solution of the following linear equation $\boldsymbol{K} = \boldsymbol{\beta} \boldsymbol{N^*}$. Therefore, for a positive steady state to be well defined, we need to assume the competition matrix $\boldsymbol{\beta}$  to be non singular, i.e., $\det (\boldsymbol{\beta}) \neq 0$.\\ 

Next, we prove that a positive steady state is globally stable if and only if the eigenvalues of the competition matrix $\boldsymbol{\beta}$ are strictly positive. Note that by definition our competition matrix $\boldsymbol{\beta}$ is symmetric, then the condition of having all eigenvalues strictly positive is equivalent to being strictly positive definite. Recall that a steady state $\boldsymbol{N^*}$ is called positive if $N_i^*>0$ for all agents $i$.

\begin{lemma}
	Consider that there exists a positive steady state, i.e., there exists $\boldsymbol{N^*}$ such that $N_i^* >0$ and $\boldsymbol K = \boldsymbol\beta \cdot \boldsymbol{N^*}$, and that the competition matrix is nonsingular. Then this steady state is asymptotically globally stable in the strictly positive quadrant $\mathbb{R}^n_{>0}$ if and only if the symmetric competition matrix $ \boldsymbol{\beta}$ is strictly positive definite.
\end{lemma}

\begin{proof}
	$\Longleftarrow$ In ref \cite{Goh}, Goh introduced a Lyapunov function that proves the global asymptotic stability in the domain $\mathbb{R}^n_{>0}$ of any positive steady state $N_i^*>0$ under the condition that the  matrix $\boldsymbol{\beta}$ is Lyapunov diagonal stable. A matrix $\boldsymbol{\beta}$ is Lyapunov diagonal stable is there exists a strictly positive diagonal matrix $\boldsymbol D$ such that $\boldsymbol D \boldsymbol{\beta} + \boldsymbol{\beta}^T \boldsymbol D$ is strictly positive definite. As in our case $\boldsymbol{\beta}$ is already strictly positive definite, then it is also Lyapunov diagonal stable. Thus, any positive steady state is globally stable. This proves the lemma from the right to the left.
		
	$\Longrightarrow$ Consider that the positive steady state $N_i^* >0$ is asymptotically globally stable. This implies that the eigenvalues of the Jacobian matrix have strictly negative real parts under the assumption that $\det(\boldsymbol{\beta})\neq 0$. The Jacobian at the positive steady state is given by the matrix $J = -D(\boldsymbol a) \boldsymbol\beta$, where $D(\boldsymbol a)$ is the diagonal matrix formed by the elements of the vector $\boldsymbol a$. The elements of $\boldsymbol a$ are strictly positive and given by $a_i = r_i / K_i N^*_i$. By similarity transformation the signature (also called the inertia) of the matrix $D(\boldsymbol a) \boldsymbol \beta$ is equal to the signature of the matrix $D(\boldsymbol a)^{1/2} \boldsymbol \beta D(\boldsymbol a)^{1/2}$. Indeed, by similarity transformations we have the following equalities:
		\begin{equation*}
		\begin{split}
			\text{signature}(D(\boldsymbol a) \boldsymbol \beta) & = \text{signature}(D(\boldsymbol a) \boldsymbol \beta D(\boldsymbol a)^{1/2} D(\boldsymbol a)^{-1/2}) \\ 
			& = \text{signature}(D(\boldsymbol a)^{1/2} \boldsymbol \beta D(\boldsymbol a)^{1/2}).
		\end{split}
		\end{equation*} 
		Moreover, as $\boldsymbol{\beta}$ is symmetric, Sylvester's law implies 
		$$\text{signature}(D(\boldsymbol a)^{1/2} \boldsymbol \beta D(\boldsymbol a)^{1/2})= \text{signature}(\boldsymbol \beta).$$ 
		Therefore the eigenvalues of $\beta$ are all strictly positive, and this proves the lemma from the right to the left.
\end{proof}

Lemma 2 implies that if we want the global asymptotic stability of a positive steady state we have to limit the level of global competition $\mu$ such that all eigenvalues of the matrix $\boldsymbol\beta$ are strictly positive. Indeed, for $\mu = 0$ the eigenvalues of the matrix $\boldsymbol\beta$ are all equal to one. As the eigenvalues are a continuous function of $\mu$, there exists a critical level $\hat{\mu}$ at which the lowest eigenvalue is equal to zero. Thus, for a level of global competition in the interval $0 \leq \mu < \hat{\mu}$, a positive steady state is asymptotically globally stable.
 
The previous lemma establishes the global asymptotic stability condition of a positive steady state. However, a positive steady state does not exist for all vectors of carrying capacity $\boldsymbol K \in \mathbb{R}^n$. There is in fact a subset of carrying capacity vectors compatible with a positive steady state. This subset is by definition $F_D = \{ \boldsymbol{K} \in \mathbb{R}^n |  \text{there exist } N^*_i >0 \text{, such that } K_i = \sum_j \beta_{ij} N^*_j \}$. That subset can simply be expressed as the strictly positive linear combination of the vectors $\boldsymbol{v_k} = \boldsymbol{\beta} \boldsymbol{e_k}$ ($\boldsymbol{e_k}$ are the vectors of the standard orthonormal basis of $\mathbb{R}^n$), $F_D = \{ \lambda_1 \boldsymbol{v_1} + \cdots + \lambda_n \boldsymbol{v_n} | \lambda_1, \cdots \lambda_n > 0\}$. As the elements of the matrix $\boldsymbol\beta$ are all positive, this implies that the vectors  $\boldsymbol{v_k}$ have all their elements positive, and in turn this also implies that the vectors of carrying capacity leading to a positive steady state have all their elements positive, i.e., $F_D \subset \mathbb{R}^n_{\geq 0}$

In the next lemmas, we study the existence and stability of steady states in the positive quadrant $\mathbb{R}^n_{\geq 0}$ for any vector of carrying capacity $\boldsymbol K$. First, let us remark that without loss of generality, we can always assume that a steady state has the following form $\boldsymbol{N^*} = ( 0, \cdots, 0, \underbrace{N^*_{m+1}, \cdots, N^*_n }_{>0}  )^{T} $.  Indeed, this form can always be achieved by renumbering the agents such that the first $m$'s are the non-positive ones and the last $n-m$ are the positive ones. 

\begin{lemma}
	Consider that the symmetric competition matrix $ \boldsymbol\beta$ is strictly positive definite. Then, for all vectors of carrying capacity $\boldsymbol K \in \mathbb{R}^n$, there exists one and only one steady state, written without loss of generality in the form $\boldsymbol{N^*} = ( 0, \cdots, 0, \underbrace{N^*_{m+1}, \cdots, N^*_n }_{>0}  )^{T} $, that is globally asymptotically stable in the domain $\Omega = R_{\geq 0}^m \cup R_{>0}^{n-m}$. Moreover, all other steady states in the positive quadrant $\mathbb{R}_{\geq 0}^n$ are unstable. Finally, the value of this stable steady state is only determined by the competition matrix $ \boldsymbol\beta$ and the carrying capacity vector $\boldsymbol K$. 
\end{lemma}
\begin{proof}

	\begin{enumerate}
	
		\item 
		Consider $\boldsymbol{N^*} = ( 0, \cdots, 0, \underbrace{N^*_{m+1}, \cdots, N^*_n }_{>0}  )^{T}$ to be a steady state. The Jacobian evaluated at this steady state is then given by the following 2-by-2 block matrix:
		$$J = -D( \boldsymbol{b}) \begin{pmatrix}
			\sum_j \beta_{1j}N^*_j - K_1 & \dots & 0 & 0 & \dots & 0  \\
			\vdots & \ddots & \vdots & \vdots & \ddots & \vdots \\
			0 & \dots & \sum_j \beta_{mj}N^*_j - K_m& 0 & \dots & 0 \\
			N^*_{m+1}\beta_{m+1,1} & \dots & N^*_{m+1}\beta_{m+1,m} & N^*_{m+1}\beta_{m+1,m+1} & \dots & N^*_{m+1}\beta_{m+1,n}  \\
			\vdots & \ddots & \vdots & \vdots & \ddots & \vdots \\
			N^*_{n}\beta_{n,1} & \dots & N^*_{n}\beta_{n,m} & N^*_{n}\beta_{n,m+1} & \dots & N^*_{n}\beta_{n,n} 
		\end{pmatrix}.$$
		The elements of the vector $\boldsymbol{b}$ are strictly positive and given by $b_i = r_i / K_i$, and the matrix $D(\boldsymbol{b})$ is a diagonal matrix formed by the elements of the vector $\boldsymbol{b}$. The steady state $\boldsymbol{N^*}$ is locally stable if and only if $\sum_j \beta_{ij}N^*_j - K_i > 0$ for all $i \in \{1,\cdots,m\}$, and the real parts of the eigenvalues of the sub-matrix 
		$$\begin{pmatrix}
			b_{m+1}N^*_{m+1}\beta_{m+1,m+1} & \dots & b_{m+1}N^*_{m+1}\beta_{m+1,n}  \\
			\vdots & \ddots & \vdots \\
			b_{n}N^*_{n}\beta_{n,m+1} & \dots & b_{n}N^*_{n}\beta_{n,n}
		\end{pmatrix} $$ 
		are strictly positive. The latter condition is automatically satisfied as the matrix $\boldsymbol{\beta}$ is symmetric and strictly positive definite. Then, the conditions of existence and local stability of $\boldsymbol{{N^*}}$ can be summarized by:
		$$N^*_i \geq 0, \quad \sum_j \beta_{ij}N^*_j - K_i \geq 0 \quad \text{ and } \quad N^*_i(\sum_j \beta_{ij}N^*_j - K_i) = 0,$$
		for all agents $i$, with the second inequality being strict if $N_i = 0$.
		
		\item
		We recall that a vector $\boldsymbol{N^*}$ is the solution of a linear complementarity problem \cite{Berman} defined by the competition matrix  $\boldsymbol{\beta}$ and the carrying capacity vector $\boldsymbol{K}$ if it satisfies the following inequalities:
		$$N^*_i \geq 0, \quad \sum_j \beta_{ij}N^*_j - K_i \geq 0 \quad \text{ and } \quad N^*_i(\sum_j \beta_{ij}N^*_j - K_i) = 0.$$
		Moreover, as in our case, the competition matrix  $\boldsymbol{\beta}$ is strictly positive definite and there exists one and only one solution to that linear complementarity problem \cite{Takeuchi}. 
						
		\item
		We prove that the steady state, which is the solution of the linear complementarity problem defined by the competition matrix  $\boldsymbol{\beta}$ and the carrying capacity vector $\boldsymbol{K}$ is asymptotically globally stable in the domain $\Omega = R_{\geq 0}^m \cup R_{>0}^{n-m}$. The proof is based on the following Lyapunov function introduced by Goh in ref. \cite{Goh2}:
		$$V(\boldsymbol{N}) = \sum_{i=1}^m d_i N_i + \sum_{i=m+1}^n d_i  \left( N_i - N_i^* + \frac{1}{N_i^*} \log \left(\frac{N_i}{N_i^*} \right) \right),$$
		with $d_i$ some strictly positive numbers.
		Clearly, we have $V(\boldsymbol{N}) \geq 0$, as $N_i^* \geq 0$, and $N_i - N_i^* + \frac{1}{N_i^*} \log \left(\frac{N_i}{N_i^*} \right) \geq 0$ for all $i \in \{m+1,\cdots,n\}$. Moreover $V(\boldsymbol{N})  =0$ if and only if $\boldsymbol{N} = \boldsymbol{N^*}$. Let us compute its derivative as a function of time. We obtain 
		$$\frac{dV}{dt} = \sum_{i=1}^m d_i \frac{r_i}{K_i} N_if_i + \sum_{i=m+1}^n d_i \frac{r_i}{K_i}(N_i - N_i^*)f_i,$$
		where $f_i = K_i - \sum_{j=1}^n \beta_{ij}N_j$. For $i \in \{m+1,\cdots,n\}$, consider the fact that $K_i = \sum_{i=1}^n \beta_{ij}N_j^*$, then we can write $f_i$ as: $f_i = -\sum_{j=1}^n \beta_{ij} (N_j-N_j^*)$. For $i \in \{1,\cdots,m\}$, we rewrite $f_i$ like: $f_i = K_i - \sum_{j=1}^n \beta_{ij}N_j^* - \sum_{i=j}^n \beta_{ij}(N_j-N_j^*)$. Substituting these two expressions into the derivative of the Lyapunov function we obtain
		$$\frac{dV}{dt} = \sum_{i=1}^m d_i \frac{r_i}{K_i}d_iN_i(K_i - \sum_{j=1}^n \beta_{ij}N_j^*) - \sum_{i=1}^n \frac{r_i}{K_i}d_iN_i(N_i - N_i^*)\beta_{ij}(N_j - N_j^*).$$
		The first term of the right side is always negative, indeed, $N_i \geq 0$ and for $i \in \{1,\cdots,m\}$ we have $K_i - \sum_{j=1}^n \beta_{ij}N_j^* \leq 0$. The second term of the right side is always strictly positive. Indeed, if we set $d_i = \frac{K_i}{r_i}$, then it is a  quadratic form defined by the strictly positive definite matrix competition matrix $\boldsymbol{\beta}$. Therefore, in the domain $\Omega$, we have that $\frac{dV}{dt} <0$. Thus, the steady state, which is the solution of the linear complementarity problem, is asymptotically globally stable in the domain $\Omega$.
		
		\item
		Consider that we have another steady state, the one given by the solution of the linear complementarity problem. Then, by the uniqueness of the solution of the linear complementarity problem, there is an agent $k$ for which $N^*_k = 0$ and at the same time $\sum_j \beta_{ij}N^*_j - K_i <0$. This implies that one eigenvalue of the Jacobian is strictly positive, thus this steady state is unstable. Therefore, there exists one and only one globally stable steady state, which is given by the solution of the linear complementarity problem defined by the competition matrix $\boldsymbol{\beta}$ and the carrying capacity vector $\boldsymbol{K}$. This proves the two first assertions of the lemma. For the last assertion it is enough to remark that the solution of the linear complementarity is only function $\boldsymbol{\beta}$ and vector $\boldsymbol{K}$. Therefore, the value of the stable steady state is also only a function of $\boldsymbol{\beta}$ and vector $\boldsymbol{K}$.
		
	\end{enumerate}
		
\end{proof}

All these lemmas together imply that under the condition that all eigenvalues of $ \boldsymbol\beta$ are strictly positive, i.e., $ \boldsymbol\beta$ is a strictly positive definite matrix, the trajectories of the dynamical system (\ref{equ:ode}) starting in the strictly positive quadrant converge to a unique steady state. Moreover, for a given competition matrix $\boldsymbol\beta$, the value of that steady state is only function of the carrying capacity $\boldsymbol{K}$; the growth rate $\boldsymbol{r}$ only dictates the velocity at which the trajectory converges to the stable steady state.   
\\

{\bf Appendix B. Alternative inter-agent resource-competition networks.} We use a resampling procedure that is able to generate a large gradient of inter-agent resource-competition networks while preserving the total number of interactions in the network\cite{Rohr}. 

First, we randomize the resource-agent system (i.e, the bipartite network) between agents (countries) and resources (companies). Note that two agents interact if they share a resource, and the strength of the interaction is equal to the number of shared resources. This randomization is performed by inferring the probability of an interaction between an agent $i$ and a resource $k$ using the model
\begin{equation} \label{equ:s:mc}
\text{logit}(p(T)_{ik}) = \frac{1}{T} \left(-\kappa (v_i - f_k)^2 + \phi_1 v_k^* + \phi_2 f_k^* \right) + m(T).
\end{equation}
The term $v_i^*$ quantifies the variability in number of resources, the term  $f_k^*$ quantifies the assortative structure of the system, and the temperature $T$ modulates the level of stochasticity in the model.  Since  $v_i^*$ and  $f_k^*$ are a priori unknown, they can be estimated from the observed resource-agent system itself. The parameters $\kappa$, $\phi_1$, and $\phi_2$ are positive scaling parameters that give the importance of the contributions of the terms. Then, based on their estimation, the probability of an interaction between all pairs of agents and resources is inferred. Thus, an alternative resource-agent system can simply be generated by drawing randomly the interactions based on those estimated interaction probabilities. The intercept $m(T)$ is adjusted for each temperature value such that the expected number of interactions is equal to the observed one. When the temperature goes to infinite, our model converges to the Erd\H{o}s-R\'{e}nyi model, when the temperature goes to zero, the system freezes in the most probable configuration predicted by our model, and when $T=1$ we recover the expected distribution of resources.

Second, we transform the previously generated resource-agent system into an inter-agent resource-competition network. This competition network is characterized by a symmetric matrix $\boldsymbol{\beta}$ of size $N \times N$, called the competition matrix. The elements of the competition matrix $\beta_{ij}$ are a function of the number of shared resources between agents.
\\ \\
{\bf ACKNOWLEDGMENTS} We thank Peter Claeys, Daniel B. Stouffer, and Brian Uzzi for valuable comments on an earlier draft of this manuscript. Funding was provided by the European Research Council through an Advanced Grant (JB), FP7-REGPOT-2010-1 program under project 264125 EcoGenes (RPR), and the Spanish Ministry of Education through a FPU PhD Fellowship (LJG).
\\ \\
{\bf Competing financial interests} The authors declare no competing financial interests.

\clearpage
\medskip

\renewcommand{\baselinestretch}{1.5}
 %
 %
 %


\clearpage


\begin{figure}[ht]
\centerline{\includegraphics*[width= 0.7 \linewidth]{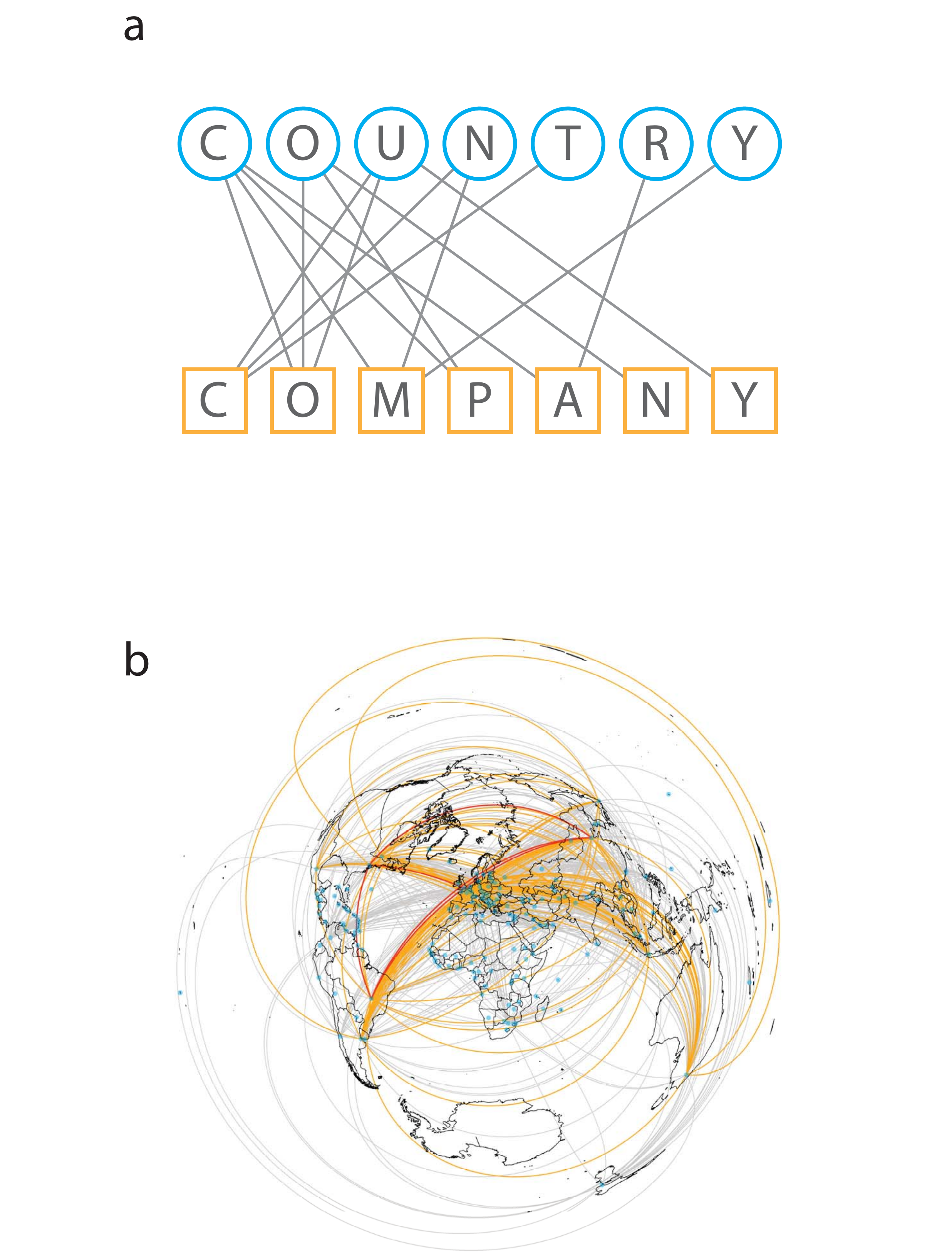}}
\caption{\small Network representation of a global socioeconomic system. The global socioeconomic network is represented by the inter-agent resource-competition network extracted from the resource-agent system. {\bf (a)} The resource-agent system is given by the interactions between agents (countries, represented by circles) and resources (companies, represented by squares). {\bf (b)} The inter-agent resource-competition network is formed by the interactions among agents competing/sharing resources and weighted by their corresponding number of shared resources. Countries are represented by their administrative capital (blue symbols), and  the darker/reddish the interaction the larger the number of companies shared. For the sake of clarity, we do not show interactions between countries that share less than 10 companies. Azimuthal equidistant projection of the Earth centered in longitude 10 and latitude 20 degrees.
}
\label{fig1}
\end{figure}

\begin{figure}[ht]
\centerline{\includegraphics*[width=1 \linewidth]{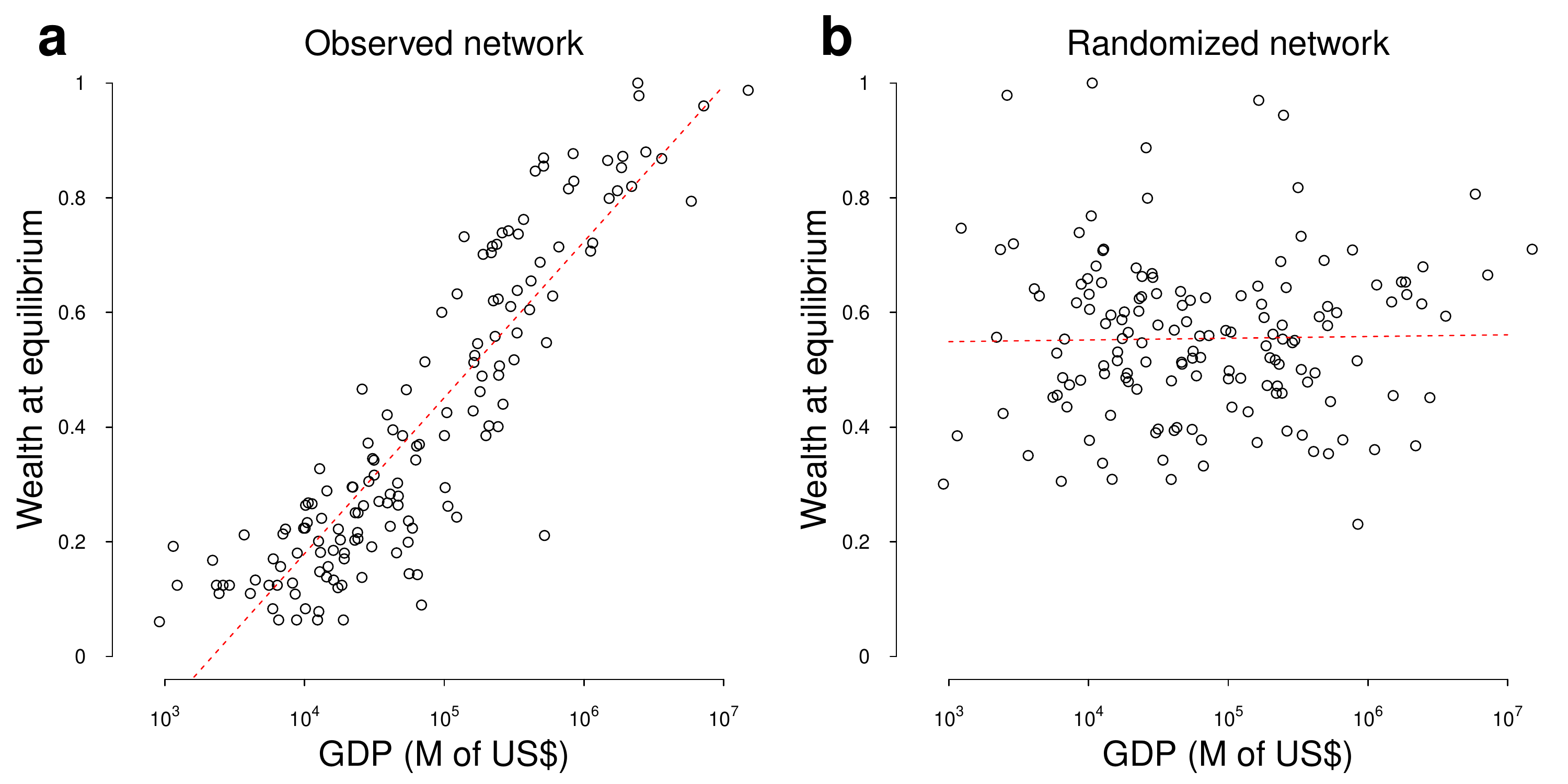}}
\caption{\small Model-generated wealth and empirical GDP. The figure shows the model-generated wealth at a stable equilibrium $N^*_i>0$ for each country (agent) and their empirical GDP in 2013. {\bf (a)} shows that wealth at equilibrium and GDP are significantly and positively correlated ($r=0.88$, Spearman rank correlation) when the dynamical model is parameterized with the structural vector of the observed resource-competition network. {\bf (b)} shows a non-significant correlation ($r=0.003$, Spearman rank correlation) when the dynamical model is parameterized by the structural vector of an alternative competition network where interactions are randomized in a similar fashion to an Erd\H{o}s-R\'{e}nyi model (Appendix B). Here, we show the results for the dynamical model using a half of the boundary of maximum global competition; however, all levels of global competition that satisfy the global stability condition yield similar results.
}
\label{fig2}
\end{figure}

\clearpage

\begin{figure}[ht]
\centerline{\includegraphics[width=0.9 \linewidth]{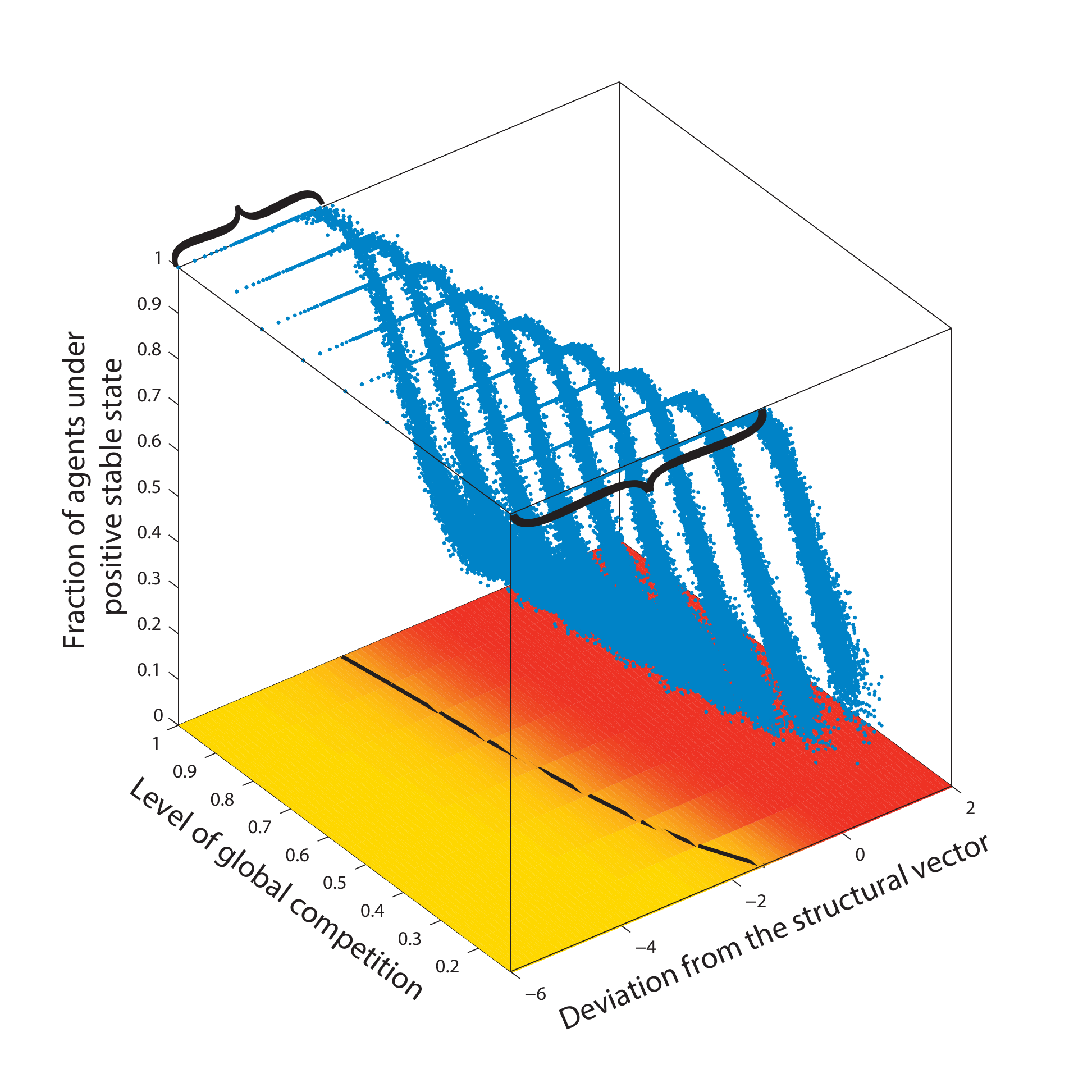}}
\caption{\small Structural stability of a global socioeconomic system. The figure presents the fraction of countries (agents) that remain in a positive stable steady state as function of both the level of deviation $\eta$ (on a log scale) from the structural vector and the level of global competition (standardized to the boundary of maximum global competition). The system is structurally stable within the parameter space compatible with all countries in a positive stable steady state ($N^*_i>0$, yellow/light region). The higher the level of global competition (black dashed line), the smaller the structural stability of the system (e.g. see brackets). For each level of global competition, we simulate different equilibrium points $N_i^*$ by initializing the model with different random proportional perturbations to the structural vector of carrying capacities.
}
\label{fig3}
\end{figure}

\begin{figure}[ht]
\centerline{\includegraphics[width=0.7 \linewidth]{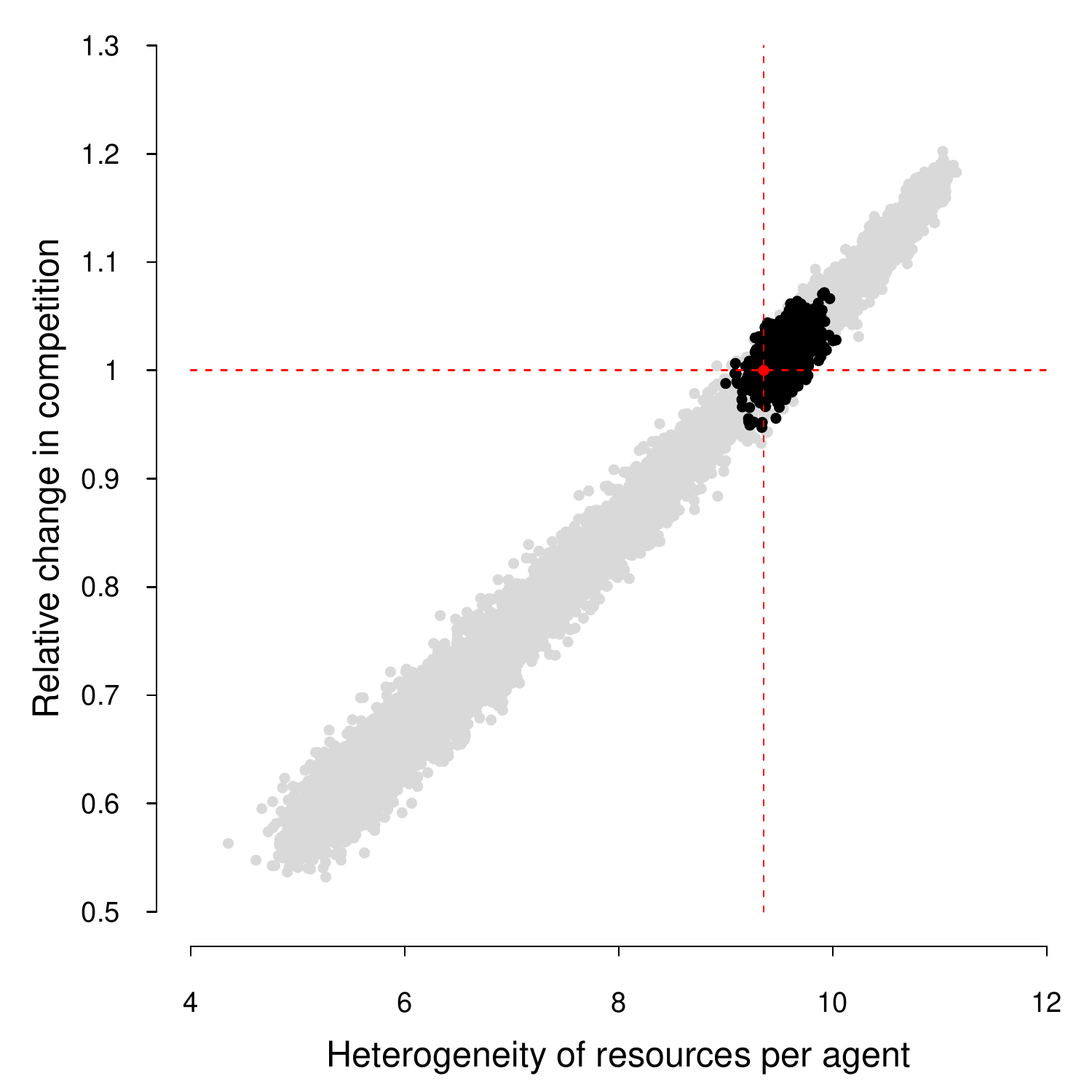}}
\caption{\small Association between distribution of resources and level of global competition. The figure shows that the higher the heterogeneity (standard deviation) in the distribution of resources, the higher the level of global competition in the inter-agent resource-competition system. The x-axis corresponds to the family of distribution of resources calculated from alternative resource-competition networks, which are extracted from randomly generated resource-agent systems (see Appendix B). The y-axis correspond to the relative change ($\rho / \rho^*$) between the level of competition in an alternative competition network $\rho$ and the level of competition in the observed competition network $\rho^*$ (red symbol). The black symbols correspond to alternative competition networks generated by preserving the expected distribution of resources per agent (Appendix B).
}
\label{fig4}
\end{figure}

\begin{figure}[ht]
\centerline{\includegraphics[width=1.2 \linewidth]{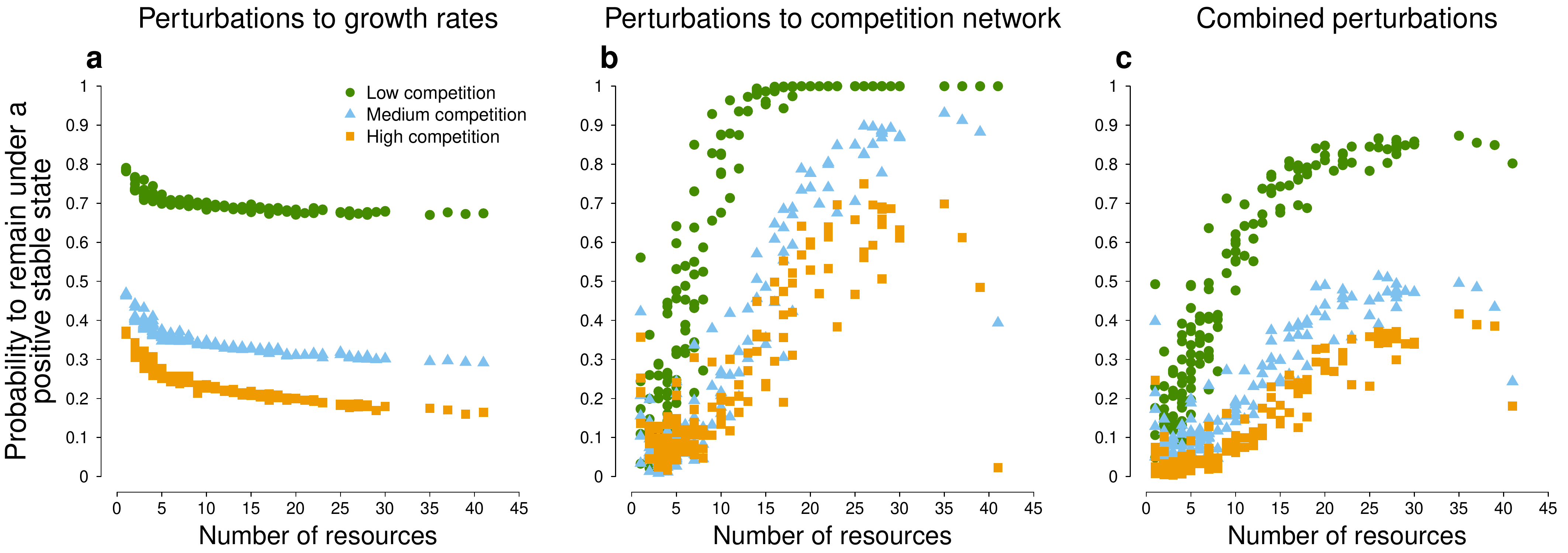}}
\caption{\small Risk assessment of changes in the global socioeconomic system. For high, medium, and low levels of global competition (colors/symbols) the figure shows as function of the number of resources the fraction of times each country (agent) remains in a positive stable steady state after ({\bf a}) a large gradient of proportional random perturbations to the structural vector of carrying capacities; ({\bf b}) changes in the resource-competition network; and ({\bf c}) to a combination of {\bf a} and {\bf b}. Each point corresponds to a country. In each scenario, we simulate 100 thousand different cases.
}
\label{fig5}
\end{figure}

\end{document}